\documentclass[journal,twoside,web]{ieeecolor}
\usepackage{lcsys}
\usepackage{cite}
\usepackage{amsmath,amssymb,amsfonts}
\usepackage{algorithmic}
\usepackage{graphicx}
\usepackage{textcomp}
\def\BibTeX{{\rm B\kern-.05em{\sc i\kern-.025em b}\kern-.08em
    T\kern-.1667em\lower.7ex\hbox{E}\kern-.125emX}}
\newboolean{long}
\setboolean{long}{true}
\usepackage{hyperref}
\usepackage{graphicx}		
\usepackage{wrapfig}
\usepackage[format=plain,font=footnotesize,labelfont=bf,labelsep=period]{caption}
\usepackage[export]{adjustbox}
\usepackage{bm}

\usepackage{amsmath}
\usepackage{mathrsfs}
\usepackage{amssymb}
\usepackage{mathtools}

\usepackage{pdfsync}
\usepackage[normalem]{ulem}
\usepackage{paralist}	
\usepackage[space]{grffile} 

\usepackage{color}

\usepackage{amsthm}
\usepackage{centernot}
\usepackage{multicol}

\newtheorem{theorem}{Theorem}
\newtheorem*{theorem*}{Theorem}
\newtheorem{corollary}{Corollary}
\newtheorem{lemma}{Lemma}
\newtheorem*{lemma*}{Lemma}
\newtheorem{definition}{Definition}
\newtheorem{assumption}{Assumption}

\newtheorem{proposition}{Proposition}
\newtheorem*{proposition*}{Proposition}

\newcommand{\R}{\mathbb{R}}

\newcommand{\C}{\mathcal{C}}

\definecolor{darkblue}{RGB}{0,0,102}
\definecolor{lightblue}{RGB}{77,77,148}

\definecolor{gold}{RGB}{234, 170, 0}
\definecolor{metallic_gold}{RGB}{139, 111, 78}

\renewcommand{\cal}[1]{\mathcal{ #1 }}

\newcommand{\mb}[1]{\mathbf{ #1 }}
\newcommand{\bs}[1]{\boldsymbol{ #1 }}

\DeclareMathOperator*{\argmin}{argmin}

\newcommand{\lmat }{\begin{bmatrix}}
\newcommand{\rmat}{\end{bmatrix}}

\newcommand{\mart}{W}

\renewcommand{\P}{\mathbb{P}} 
\newcommand{\E}{\mathbb{E}} 
\newcommand{\F}{\mathscr{F}} 
\usepackage{dsfont}

\usepackage{svg}

\newcommand{\revised}[1]{\textcolor{black}{#1}}

\begin{document}

\title{\LARGE{Bounding Stochastic Safety: Leveraging Freedman's Inequality with Discrete-Time Control Barrier Functions}}
\author{Ryan K. Cosner, Preston Culbertson, and  Aaron D. Ames
\thanks{The authors are with the Department of Mechanical and Civil Engineering at the California Institute of Technology, Pasadena, CA 91125, USA.   
        {\tt\small\{rkcosner, pbulbert, ames\}@caltech.edu} 
        This work was supported by BP and NSF CPS Award \#1932091. }
\ifthenelse{\boolean{long}}{
    }{
        \vspace{-10em}
    }
}

\pagestyle{empty} 
\maketitle
\thispagestyle{empty}

\vspace{-5em}
\begin{abstract}
When deployed in the real world, safe control methods must be robust to unstructured uncertainties such as modeling error and external disturbances. Typical robust safety methods achieve their guarantees by always assuming that the worst-case disturbance will occur. 
In contrast, this paper utilizes Freedman's inequality in the context of discrete-time control barrier functions (DTCBFs) and c-martingales to provide stronger (less conservative) safety guarantees for stochastic systems. Our approach accounts for the underlying disturbance distribution instead of relying exclusively on its worst-case bound and does not require the barrier function to be upper-bounded, which makes the resulting safety probability bounds more useful for intuitive safety constraints such as signed distance. We compare our results with existing safety guarantees, such as input-to-state safety (ISSf) and martingale results that rely on Ville's inequality. When the assumptions for all methods hold, we provide a range of parameters for which our guarantee is stronger. Finally, we present simulation examples, including a bipedal walking robot, that demonstrate the utility and tightness of our safety guarantee.   
\vspace{0em}
\end{abstract}



\ifthenelse{\boolean{long}}{
    }{
        \vspace{-0.3em}
    }
\section{Introduction}

\IEEEPARstart{S}{afety}---typically characterized as the forward-invariance of a safe set \cite{ames_control_2017}---has become a popular area of study within control theory, with broad applications to autonomous vehicles, medical and assistive robotics, aerospace systems, and beyond. Ensuring safety for these systems requires one to account for unpredictable, real-world effects. 
Historically, control theory has treated the problem of safety under uncertainty using deterministic methods, often seeking safety guarantees in the presence of bounded disturbances. This problem has been studied using a variety of safe control approaches including control barrier functions (CBFs) \cite{kolathaya_input--state_2019}, backwards Hamilton-Jacobi (HJ) reachability \cite{bansal_hamilton-jacobi_2017}, and state-constrained model-predictive control (MPC) \cite{borrelli2017predictive}. However, this worst-case analysis often leads to conservative performance since it ensures robustness to adversarial disturbances which are uncommon in practice.

\begin{figure}[t]
    \centering
    \includegraphics[width=\linewidth]{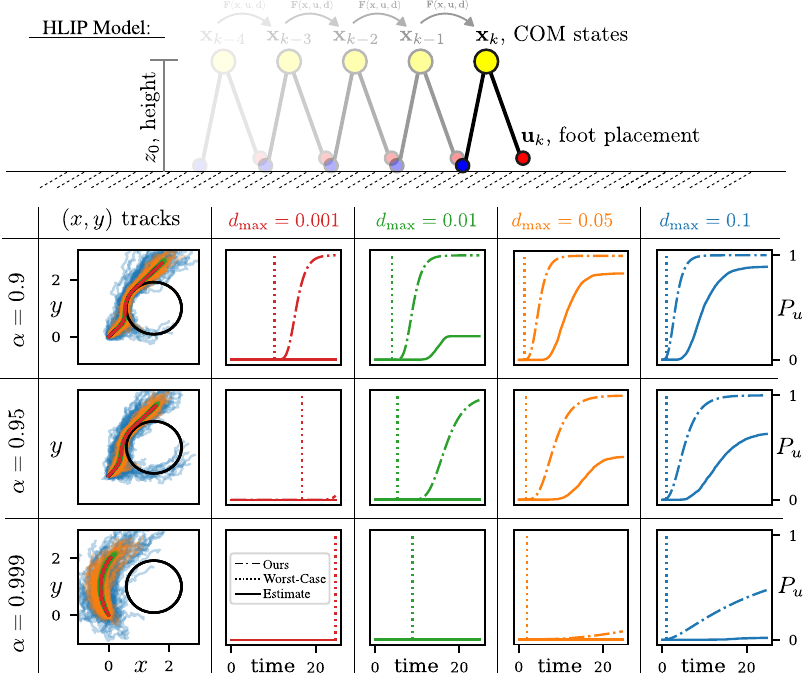}
    \caption{Safety results for a bipedal robot navigating around an obstacle using our method. Details are provided in Section \ref{example:hlip}. \textbf{(Top)} Visualization of the Hybrid Linear Inverted Pendulum (HLIP) model. Yellow indicates the center-of-mass (COM), blue is the stance foot, and red is the swing foot. The states $\mb{x}_k$ are the global COM position, the relative COM position, and COM velocity, and the input is the relative position of the feet at impact. 
    \textbf{(Bottom)} A table with variable maximum disturbance value $(d_\textup{max})$ and controller parameter ($\alpha$) shows our \revised{(dashed lines)} theoretical bound on safety failure from Thm. \ref{thm:main}\revised{, (dotted lines) the shortest first-violation time based on the worst-case disturbance approximation,} and \revised{(solid lines)} approximated probabilities from 5000 trials (lower is \revised{safer}). On the left, the trajectories of the COM are shown \revised{walking from bottom left towards the top right while avoiding the obstacle} with each color corresponding to a different $d_\textup{max}$. The robot attempts to avoid the obstacle \revised{(black)}. Code to reproduce this plot can be found at \cite{codebase}.}
    \label{fig:hlip}
\ifthenelse{\boolean{long}}{
    }{
        \vspace{-2em}
    }

\end{figure}

Stochastic methods provide an alternative to the worst-case bounding approach. Instead of a conservative uncertainty bound, these methods consider a distribution of possible disturbances. Although they do not provide absolute, risk-free safety guarantees, they allow for smooth degradation of safety via variable, risk-aware levels of conservatism. 
\revised{A wide variety of stochastic safety methods exist including: reachability-based optimal safety \cite{abate2008probabilistic, franzle2011measurability}, constrained coherent risk measures \cite{chapman2021risk}, sampling-based general risk measures \cite{lindemann2021stl}, and martingale-based methods \cite{steinhardt_finite-time_2012, santoyo_barrier_2021} amongst many others. In this work, we will focus on martingale-based methods due to their ability to generate trajectory-long guarantees and their relative simplicity as a method which relies primarily on only a distribution's first-moment.}

\revised{Continuous-time martingale-based stochastic safety methods have successfully achieved} strong probabilistic safety guarantees \cite{so2023almostsure, Black_2023, kushner_stochastic_1967}\revised{, \cite{prajna2004stochastic}}, but generally require controllers with functionally infinite bandwidth, a strong assumption for real-world systems with discrete-time sensing and actuation. Alternatively, discrete-time methods have shown success while also capturing the sampled-data complexities of most real-world systems \cite{cosner_robust_2023,santoyo_verification_2019,steinhardt_finite-time_2012, mathiesen2023inner}. In this work we focus on extending the theory \revised{of} discrete-time \revised{martingale-based} stochastic safety involving discrete-time control barrier functions (DTCBFs) and $c$-martingales.



The stochastic \revised{distrete-time martingale-based stochastic safety literature has shown significant theoretical success \cite{santoyo_barrier_2021,mathiesen2023inner, cosner_robust_2023,steinhardt_finite-time_2012,kushner_stochastic_1967} in generating risk-based safety guarantees and in deploying these guarantees to real-world systems \cite{cosner2023generative}}. 
\revised{In this work we seek to extend these} existing martingale-based safety techniques by utilizing a \revised{different (and often stronger)} concentration inequality that can provide sharper safety probability bounds. Where other works 
have traditionally relied on Ville's inequality \cite{ville1939etude}, 
we instead turn to Freedman's inequality \cite{freedman1975tail}. 
By additionally assuming that the martingale \revised{differences} and predictable quadratic variation are bounded, this inequality relaxes the \revised{nonnegativity} assumption required by Ville's \revised{inequality} while also providing generally tighter bounds that degrade smoothly with increasing uncertainty.




This paper combines discrete-time martingale-based safety techniques with Freedman's inequality to obtain tighter bounds on stochastic safety. We make three key contributions: (1) introducing Freedman-based safety probabilities for DTCBFs and $c$-martingales, (2) providing a range of parameter values where our bound is tighter than existing \revised{discrete-time martingale-based safety} results, and (3) validating our method in simulation.  We apply our results to a bipedal obstacle avoidance scenario (Fig. \ref{fig:hlip}), using a reduced-order model of the step-to-step dynamics. 
This case study shows the utility of our probability bounds, which decay smoothly with increasing uncertainty and enable non-conservative, stochastic collision avoidance for bipedal locomotion. 



\ifthenelse{\boolean{long}}{
    }{
        \vspace{-0.0em}
    }

\section{Background}

Let $(\Omega, \F, \P)$ be a probability space and let \revised{$\F_0 \subset \F_1 \subset \dots \subset \F$} be a filtration of $\F$. Consider discrete-time dynamical systems of the form: 
\ifthenelse{\boolean{long}}{
    }{
        \vspace{-0.3em}
    }

\begin{align}
\mb{x}_{k+1} = \mb{F}(\mb{x}_k, \mb{u}_k, \mb{d}_k), \quad \forall k \in \mathbb{Z} \label{eq:ol_dyn}
\end{align}
\ifthenelse{\boolean{long}}{
    }{
        \vspace{-0.3em}
    }
where $\mb{x}_k \in \R^n $ is the state, $\mb{u}_k \in \R^m $ is the input, $\mb{d}_k$ is an $\F_{k+1}$ measurable random disturbance which takes values in $\R^\ell$
, and $\mb{F}: \R^n \times \R^m \times \R^\ell  \to \R^n$ is the dynamics. Throughout this work we assume that all random variables and functions of random variables are integrable. 

To create a closed-loop system, we add a state-feedback controller $\mb{k}: \R^n \to \R^m $: 
\begin{align}
    \mb{x}_{k+1} = \mb{F}(\mb{x}_k, \mb{k}(\mb{x}_k), \mb{d}_k), \quad \forall k \in \mathbb{Z} \label{eq:cl_dyn}
\end{align}

\noindent The goal of this work is to provide probabilisic safety guarantees for this closed-loop system.

\ifthenelse{\boolean{long}}{
    }{
        \vspace{-1em}
    }
\subsection{Safety and Discrete-Time Control Barrier Functions}

To make guarantees regarding the safety of system \eqref{eq:cl_dyn}, we first formalize our notion of safety as the forward invariance of a \revised{user-defined} ``safe set'', $\cal{C}\subset\R^n$, as is common in the robotics and control literature \cite{ames_control_2017,  bansal_hamilton-jacobi_2017,borrelli2017predictive, khatib1986real}. 

\begin{definition}[Forward Invariance and Safety\footnote{\revised{
For this work we will focus on the safety of \eqref{eq:cl_dyn} exclusively at samples times as in \cite{borrelli2017predictive} and \cite{agrawal_discrete_2017}. We refer to \cite{breeden_control_2022} for an analysis of intersample safety.}}] \label{def:safety}
    A set $\C \subset\R^n $ is forward invariant for system \eqref{eq:cl_dyn} if $\mb{x}_0 \in \C \implies \mb{x}_k \in \C $ for all $k \in \mathbb{N}$. We define ``safety'' \revised{with respect to $\mathcal{C}$} as the forward invariance of $\C$.  
\end{definition} 

One method for ensuring safety is through the use of Discrete-Time Control Barrier Functions (DTCBFs). For DTCBFs, we consider safe sets that are $0$-superlevel sets \cite{ames_control_2017} of a continuous function $h:\R^n \to \R$:  
\begin{align}
    \C = \{ \mb{x} \in \R^n ~|~ h(\mb{x}) \geq 0 \}. 
\end{align}

In particular the DTCBF is defined as:
\begin{definition}[Discrete-Time Control Barrier Function (DT-CBF) \cite{agrawal_discrete_2017}]
Let $\C \subset \R^n$ be the $0$-superlevel set of some function $h: \R^n \to \R$. The function $h$ is a DTCBF for $\mb{x}_{k+1} = \mb{F}(\mb{x}_k, \mb{u}, \mb{0}) $ if there exists an $\alpha \in [0,1]$ such that: 
\begin{align}
    \sup_{\mb{u} \in \R^m} h(\mathbf{F}(\mathbf{x}, \mathbf{u}, \mb{0})) > \alpha h(\mb{x}), \quad \quad \forall \mathbf{x}\in\C \label{eq:det_dtcbf_cond}
\end{align} 
\end{definition}

DTCBFs differ from their continuous-time counterparts in that they satisfy an inequality constraint on their \textit{finite difference} instead of their derivative\footnote{The standard continuous-time CBF condition $\dot{h}(\mb{x}) \leq -\overline{\gamma} h(\mb{x})$ for $\overline{\gamma}>0$ becomes $  h(\mb{x}_{k+1}) - h(\mb{x}_k)  \geq - \gamma h(\mb{x}_k)$ for $\gamma \in [0,1]$ in discrete-time; defining $\alpha = 1 - \gamma$ recovers the condition $h(\mb{x}_{k+1}) \geq \alpha h(\mb{x}_k)$.}. 
On the other hand, they are similar in their ability to create \emph{safety filters} for nominal controllers $\mb{k}_\textup{nom}: \R^n \times \mathbb{Z} \to \R^m$ of the form:
\begin{align}
    \mb{k}(\mb{x}) = \argmin_{\mb{u} \in \R^m } & \quad \Vert \mb{u} - \mb{k}_\textup{nom}(\mb{x}, k) \Vert^2 \label{eq:dtcbf-op}\\
    \textup{s.t. } & \quad h(\mb{F}(\mb{x}, \mb{u}, \mb{0})) \geq \alpha h(\mb{x}). \nonumber 
\end{align}
Assuming feasibility,
\ifthenelse{\boolean{long}}{
\revised{\footnote{If infeasible, a slack variable can be added to recover feasibility and its \revised{e}ffect on safety can be analyzed using the ISSf framework \cite{kolathaya_input--state_2019}. Additionally, unlike the affine inequality constraint that arises with continuous-time CBFs \cite{ames_control_2017}, the optimization problem \eqref{eq:dtcbf-op} is not necessarily convex. To ameliorate this issue, it is often assumed that $h\circ \mb{F}$ is concave with respect to $\mb{u}$ \cite{zeng_safety-critical_2021, agrawal_discrete_2017, ahmadi_safe_2019}.} }}{}
$\mb{k}(\mb{x})$ guarantees safety for the undisturbed system by selecting inputs that satisfy \eqref{eq:det_dtcbf_cond}\cite[Prop. 1]{agrawal_discrete_2017}.

For deterministic systems, infinite-horizon safety guarantees are common. 
However, \revised{for discrete-time stochastic systems, when the disturbance is bounded, infinite horizon guarantees fail to capture the nuances of variable risk levels and, when the disturbance is unbounded, infinite-horizon guarantees can be impossible to achieve}\footnote{\revised{Consider the system: $\mb{x}_{k+1} = \mb{u}_k + \mb{d}_k$, where $\mb{x}\in \R$, $\mb{u}\in\R$, $\mb{d}\sim \mathcal{N}(0,1)$, and $\mathcal{C}=\{ \mb{x}\in\R ~|~ \vert \mb{x} \vert < 1 $\}. At every time step, $\mathbb{P}\{ \mb{x}_{k+1} \in \mathcal{C}\}$ is maximized with $\mb{u}_k =0$, but then even over a single discrete step, there is at least 30\% chance of failure. As time continues, this constant risk of failure at every step makes infinite horizon guarantees impossible to achieve.}} \cite[Sec. IV]{culbertson_input--state_2023}. \revised{In order to provide an achievable risk-based guarantee we choose to analyze finite-time safety probabilities as in \cite{kushner_stochastic_1967, steinhardt_finite-time_2012, santoyo_barrier_2021, prajna2004stochastic} instead of infinite-time safety guarantees.} 

\begin{definition}[$K$-step Exit Probability]
    For any $K\in \mathbb{N}_1$ and initial condition $\mb{x}_0 \in \R^n$, the $K$-step exit probability of the set $\C$ for the closed-loop system \eqref{eq:cl_dyn} is: 
    \begin{align}
        P_u(K, \mb{x}_0) = \P \left\{ \mb{x}_k \notin \C \textup{ for some }  k \revised{\leq K} \right\} 
    \end{align}
\end{definition}
\noindent This describes the probability that the system will leave the safe set $ \C$ within $K $ time steps given that it started at $\mb{x}_0$.

\ifthenelse{\boolean{long}}{
    }{
        \vspace{-0.5em}
    }
\subsection{Existing Martingale-based Safety Methods}

In this work, we will generate bounds on $K$-step exit probabilities using martingale-based concetration inequalities. 
Martingales are a class of stochastic processes which satisfy a relationship between their mean and previous value. 
\begin{definition}[Martingale \cite{grimmett_probability_2020}, \cite{steinhardt_finite-time_2012}]
    Let $(\Omega, \F, \P)$ be a probability space with a filtration $\{ \F_0, \F_1,  \dots, \F\} $. A stochastic process $\mart_k $ that is adapted to the filtration and is integrable at each $k$ 
    is a martingale if 
    \begin{align}
        \E [~\mart_{k+1}~|~ \F_k~] = \mart_k, \quad \forall k \in \mathbb{Z} \quad \textup{ (a.s.\revised{)}}
    \end{align}
    Additionally, if $\mart_k$ satisfies: 
    \begin{align}
        \E[~\mart_{k+1} ~|~ \F_k~] \leq \mart_k + \revised{c}, \quad \forall k \in \mathbb{Z}\quad  \textup{ (a.s.)}, \label{eq:supermart}
    \end{align}
    \revised{with $c=0$ then it is a supermartingale and if it satisfies \eqref{eq:supermart} with $c\geq0 $ then it is a $c-$martingale.}
\end{definition}

Many concentration inequalities can be used to bound the spread of a martingale over time. One particularly useful bound is Ville's 
\cite{ville1939etude} which bounds the probability that a supermartingale $\mart_k$ rises above a threshold $\lambda>0$. 
\begin{lemma}[Ville's Inequality \cite{ville1939etude}] \label{thm:villes}
    If $\mart_k$ \revised{is} a nonnegative supermartingale, then for all $\lambda > 0$,
    \begin{align}
        \textstyle \P \left\{ \sup_{k\in \mathbb{Z}} \mart_k  > \lambda \right\} \leq \frac{\E[\mart_0]}{\lambda} \label{eq:villes}
    \end{align}
\end{lemma}
\noindent Critically, Ville's inequality assumes \textit{nonnegativity} which manifests as a requirement that $h$ be upper-bounded, e.g. \eqref{assp:upper_bound}. \ifthenelse{\boolean{long}}{
        A proof of Ville's inequality can be found in Appendix \ref{appdx:villes}
    }{
    
    }


For safety applications of Ville's inequality, we consider the case where $h(\mb{x}_k)$ is upper bounded by $B>0$ and satisfies one of the following expectation conditions
\begin{align}
    \E[ ~h(\mb{F}(\mb{x}_k, \mb{k}(\mb{x}_k), \mb{d}_k))  ~|~ \F_{k} ~] & \geq \alpha h(\mb{x}_k),  \label{eq:expect_cbf_cond} \tag{DTCBF}\\ 
        \E[~h(\mb{F}(\mb{x}_k, \mb{k}(\mb{x}_k), \mb{d}_k)) ~|~ \F_{k} ~]& \geq h(\mb{x}_k) - c \tag{$c$-mart.},\label{eq:c_mart_cond} 
\end{align}
for some $\alpha \in (0,1)$ or $c \geq 0 $. 
In this case, we can achieve the following bound on the $K$-step exit probability, $P_u(K, \mb{x}_0)$: 
\begin{theorem}[Safety using Ville's Inequality\footnote{\revised{
\ifthenelse{\boolean{long}}{
        See Appx. \ref{appx:santoro}
    }{
        See \cite[Appx. C]{cosner2024bounding} 
    }
for a discussion notational differences between this presentation of Thm. \ref{thm:rss} and that in \cite{santoyo_verification_2019} and \cite{steinhardt_finite-time_2012}. Also, see \cite[Thm. 5]{cosner_robust_2023} for probability bounds associated with the general condition $\E[h(\mb{F}(\mb{x}_k, \mb{k}(\mb{x}_k), \mb{d}_k))~|~ \mathscr{F}_k] \geq \alpha h(\mb{x}_k) - c$ using Ville's inequality. }}, \cite{cosner_robust_2023, steinhardt_finite-time_2012, kushner_stochastic_1967, santoyo_barrier_2021}] \label{thm:rss}
    If, for some $B > 0$ and $K \in \mathbb{N}_1$, the function $h: \R^n \to \R$ satisfies: 
    \begin{align}
        & \quad \quad \quad \quad \quad\quad  h(\mb{x}) \leq B, \quad \textup{for all } \mb{x} \in  \R^n, \label{assp:upper_bound} \\
       &\textup{then: } \quad\quad \quad \quad \quad \quad P_u(K, \mb{x}_0) \leq 1 -  \frac{\lambda}{B},  \label{eq:villes_safety_bound} \quad \quad \\
        &\textup{where } \lambda = \begin{cases}
            \alpha^Kh(\mb{x}_0), & \textup{if }\eqref{eq:cl_dyn} \textup{ satisfies }  \eqref{eq:expect_cbf_cond} \; \forall k \leq K\\
            h(\mb{x}_0) - cK, & \textup{if }\eqref{eq:cl_dyn} \textup{ satisfies } \eqref{eq:c_mart_cond}\;  \forall k \leq K . 
        \end{cases} \nonumber
    \end{align}
\end{theorem}
\noindent This guarantees that the risk of the becoming unsafe is upper bounded by a function which decays to 1 with time and which depends on the system's initial safety ``fraction'', $h(\mb{x}_0)/B$. 
\ifthenelse{\boolean{long}}{
        A proof of this Theorem can be found in Appendix \ref{appdx:rss_and_stein}.
    }{
    
    }

\ifthenelse{\boolean{long}}{
    }{
    }

\section{Safety Guarantees using Freedman's Inequality}

This section presents our main result: $K$-step exit probability bounds for DTCBFs and c-martingales generated using Freedman's inequality, a particularly strong martingale concentration inequality. 
Here, we use the simpler, historical version as presented by Freedman \cite{freedman1975tail}; see \cite{fan2012hoeffding} for historical context and a new, tighter alternative which could also be used. After presenting this result, this section explores comparisons with existing Ville's-based methods and input-to-state safety. 


Before presenting Freedman's inequality, we must define the predictable quadratic variation (PQV) of a process which is a generalization of variance for stochastic processes. 
\begin{definition}[Predictable Quadratic Variation (PQV) \cite{grimmett_probability_2020}]
    The PQV of a martingale $\mart_k$ at $K \in \mathbb{N}_1$ is: 
    \begin{align}
        \langle \mart \rangle_K \triangleq \textstyle \sum_{i=1}^K \E[(\mart_i - \mart_{i-1})^2 ~|~ \F_{i-1}]  
    \end{align}
\end{definition}

Unlike Ville's inequality, Freedman's inequality does not require nonnegativity of the martingale $\mart_k$, thus removing the upper-bound requirement \eqref{assp:upper_bound} on $h$. In place of nonnegativity, we require two alternative assumptions:
\begin{assumption}[Upper-Bounded Differences]\label{assp:bounded_mart_diff} 
    We assume that the martingale differences are upper-bounded by $ 1$ \revised{(i.e. $W_{k+1} - W_k \leq 1$, similar to Azuma-Hoeffding methods \cite{grimmett_probability_2020}).}
\end{assumption}
\begin{assumption}[Bounded PQV]\label{assp:bounded_pred_quad_var}
    We assume that the PQV
    is upper-bounded by $\xi^2>0$.
\end{assumption}

Given the PQV of the process, Freedman's inequality
provides the following bound:
\begin{theorem}[{Freedman's Inequality \cite[\revised{Thm. 4.1}]{freedman1975tail}}] \label{thm:azuma_hoeffding}
    If, for some $K \in \mathbb{N}_1$ and $\xi > 0 $, $\mart_k$ is a supermartingale with $W_0=0$ such that: 
    \begin{align}
        (\mart_k - \mart_{k-1})\leq 1&  \quad \textup{ for all } k \leq K,   \tag{Assumption \ref{assp:bounded_mart_diff}}\\
        \langle \mart\rangle_K \leq \xi^2,&  \tag{Assumption \ref{assp:bounded_pred_quad_var}}
    \end{align}
     then, for any $\lambda \geq 0 $, 
    \begin{align}
        \textstyle \P \left\{ \max_{k \leq K } \mart_k \geq \lambda  \right\} \leq H(\lambda, \xi) \triangleq \left( \frac{\xi^2}{\lambda + \xi^2}\right)^{\lambda + \xi^2} e^{\lambda }.   \label{eq:hoeff_ineq}
    \end{align}
\end{theorem}




\ifthenelse{\boolean{long}}{

    \noindent See  \cite[Appx. ]{cosner2024bounding} \ref{appdx:freedman} for a restatement of Freedman's proof.

}{

    \revised{To prove this, we first bound the moment generating function $\E[e^{\gamma X}] \leq e^{(e^\gamma - 1 - \gamma) \textup{Var}(X)}$ for $\gamma \geq 0 $ and 0-mean  random variables $X$ bounded by 1. By choosing $X$ as the martingale differences, this is used to prove that $ \int_{\{\tau < \infty\}} e^{\gamma W_\tau   - (e^\gamma -1 - \gamma )\langle W \rangle_\tau } d \mathbb{P} \leq 1$ for $W_k$ satisfying assp. \eqref{assp:bounded_mart_diff} where $\tau$ any stopping time \cite[Prop. 3.3]{freedman1975tail}. Bounding with Assp. \eqref{assp:bounded_pred_quad_var} and optimizing $\gamma$ then yields the desired bound. A restatement of this proof can be found in \cite[Appx. D]{cosner2024bounding}. }

}

\ifthenelse{\boolean{long}}{
    }{
        \vspace{-0.75em}
    }

\subsection{Main Result: Freedman's Inequality for Safety}


Next we present the key contribution of this paper: the application of Freedman's inequality
to systems which satisfy the \ref{eq:expect_cbf_cond}  or $c$-martingale conditions.

\begin{theorem} \label{thm:main}
    If, for some $K \in \mathbb{N}_1, \sigma> 0  $, and $\delta > 0 $, the following bounds\footnote{\revised{Only upper-bounds on $\delta$ and $\sigma^2$ are required for  \eqref{eq:main_bound} to hold and this guarantee is robust to changes in distribution that still satisfy \eqref{assp:thm_main_bounded_diff} and \eqref{assp:thm_main_pqv}. For real-world systems, distribution-learning can be employed, similar to  \cite{cosner2023generative}.}} on the difference\footnote{\revised{\ifthenelse{\boolean{long}}{
            See Appx. \ref{appx:construct_delta}
        }{
            See \cite[Appx. G]{cosner2024bounding} 
        }
        for a constructive method for determining $\delta$ and $\sigma$. 
        }
    } between the true and predictable update \eqref{assp:thm_main_bounded_diff} and the conditional variance \eqref{assp:thm_main_pqv}  hold for all $k \leq K$: 
    \begin{align}
        \E[~h(\mb{x}_{k})~|~\F_{k-1}~] - h(\mb{x}_{k})   &\leq  \delta, \label{assp:thm_main_bounded_diff} \\
        \textup{Var}(~h(\mb{x}_{k+1}) ~|~ \F_k~) & \leq \sigma^2, \label{assp:thm_main_pqv}
    \end{align}
    then the $K$-step exit probability is bounded as: 
    \begin{align}
         &\textstyle \quad \quad \quad \quad P_u(K, \mb{x}_0) \leq H\left(\frac{\lambda}{\delta}, \frac{\sigma \sqrt{K} }{ \delta }  \right), \label{eq:main_bound}\\ 
    \textup{where } \lambda &= \begin{cases}
            \alpha^Kh(\mb{x}_0), & \textup{if }\eqref{eq:cl_dyn} \textup{ satisfies }  \eqref{eq:expect_cbf_cond} \; \forall k \leq K,\\
            h(\mb{x}_0) - cK, & \textup{if }\eqref{eq:cl_dyn} \textup{ satisfies } \eqref{eq:c_mart_cond}\;  \forall k \leq K . 
        \end{cases} \nonumber
    \end{align}
\end{theorem}

\revised{
To apply Thm. \ref{thm:azuma_hoeffding} to achieve Thm. \ref{thm:main} we follow this proof structure: 
\textbf{(Step 1)} normalize $h$ and use it to construct a candidate supermartingale $W_k$, \textbf{(Step 2)} verify that $W_k$ is indeed a supermartingale with $W_0 = 0$, \textbf{(Step 3)} use Doob's decomposition \cite[Thm 12.1.10]{grimmett_probability_2020} to produce a martingale $M_k$ from $W_k$ in order to remove the negative effect of safe, predictable jumps from the PQV, \textbf{(Step 4)} verify that $M_k$ satisfies Assp.s \ref{assp:bounded_mart_diff} and \ref{assp:bounded_pred_quad_var}, \textbf{(Step 5)} choose $\lambda\geq 0$ such that a safety failure implies $\{\max_{k \leq K} W_k \geq \lambda\}$ as in \eqref{eq:hoeff_ineq}, 
and \textbf{(Step 6)} specialize to specific values of $\alpha$ and $c$ for each case.
}

\begin{proof}
    \revised{\underline{\textbf{(Step 1)}}} Consider the case, for $\Tilde{\alpha} \in (0,1]$ and $\Tilde{c}\geq 0$, 
    \begin{align}
        &\textup{where } \E[h(\mb{x}_{k+1}) | \F_k] \geq \Tilde{\alpha} h(\mb{x}_k) - \Tilde{c}, \;\; \textup{for all }  k \leq K . \label{eq:general_expect_cond}
    \end{align}
    
    First, define the normalized safety function $
        \eta(\mb{x})  \triangleq  \textstyle \frac{h(\mb{x})}{\delta}$ to ensure that the martingale differences will be bounded by 1. 
    Next, use $\eta$ to define the candidate supermartingale\footnote{ \revised{
    We use the ``empty sum'' convention that $\sum_{i=1}^0 \rho =0 $ for any $\rho\in \R$. }}
    \begin{align}
        \textstyle \mart_k \triangleq -\Tilde{\alpha}^{K - k} \eta(\mb{x}_k) + \Tilde{\alpha}^K \eta(\mb{x}_0) - \sum_{i=1}^k \Tilde{\alpha}^{K-i}\frac{\Tilde{c}}{\delta} \label{eq:cand_mart}
    \end{align} \revised{\underline{\textbf{(Step 2)}}} This  satisfies\footnote{$W_0 = 0$ since $\mb{x}_0$ is known and randomness first enters through $\mb{d}_0$.} $\mart_0 = 0$ and is a supermartingale:
    \begin{align}
        &\E[\mart_{k+1} | \F_{k}]\\ 
        &= \textstyle - \Tilde{\alpha}^{K-(k+1)} \E[\eta(\mb{x}_{k+1})|\F_k] + \Tilde{\alpha}^K \eta(\mb{x}_0) - \sum_{i=1}^{k+1} \Tilde{\alpha}^{K-i}\frac{\Tilde{c}}{\delta}, \nonumber\\
        & \leq -\Tilde{\alpha}^{K-k} \eta(\mb{x}_k)  + \Tilde{\alpha}^K \eta(\mb{x}_0)  -  \textstyle \sum_{i=1}^{k} \Tilde{\alpha}^{K-i}\frac{\Tilde{c}}{\delta} = \mart_k.  \nonumber    
    \end{align}
    \noindent which can be seen by applying the bound from \eqref{eq:general_expect_cond}.
    
    \noindent \revised{\underline{\textbf{(Step 3)}}} The martingale from Doob's decomposition is: 
    \begin{align}
        & M_k  \triangleq \mart_k + \textstyle \sum_{i=1}^k ( \mart_{i-1} - \E[\mart_i |\F_{i-1}]), \label{eq:doobs_mart} \\
        & = \mart_k + \textstyle \sum_{i=1}^k \underbrace{\textstyle\frac{\Tilde{\alpha}^{K-i}}{\delta}(\E[h(\mb{x}_{i}) | \F_{i-1}] - \Tilde{\alpha} h(\mb{x}_{i-1}) + \Tilde{c}) }_{\geq 0 } \geq \mart_k \nonumber  
    \end{align}
    where the bound 
    comes from 
    \eqref{eq:general_expect_cond} and positivity of $\Tilde{\alpha}$ and $\delta$.

    \noindent \revised{\underline{\textbf{(Step 4)}}} Furthermore, $M_k$ satisfies Assp. \ref{assp:bounded_mart_diff}:
    \begin{align}
        & M_k - M_{k-1}  = \mart_k - \E[\mart_{k}| \F_{k-1}],\\
        & = \Tilde{\alpha}^{K-k} (\E[\eta(\mb{x}_k) | \F_{k-1}] - \eta(\mb{x}_k) ) \leq \Tilde{\alpha}^{K-k} \textstyle \frac{\delta}{\delta} \leq 1,
    \end{align}
    since we assume in \eqref{assp:thm_main_bounded_diff} that $ \E[h(\mb{x}_k) ~|~ \F_{k-1}]  - h(\mb{x}_k) \leq \delta$. 

    Next,  $\Tilde{\alpha} \in (0,1] $ and   \eqref{assp:thm_main_pqv} ensure that $M_k$ satsifes Assp. \ref{assp:bounded_pred_quad_var}: 
    \begin{align}
        & \langle M \rangle_K  = \textstyle \sum_{i=1}^K\E[\Tilde{\alpha}^{2(K-i)}(\eta(\mb{x}_i) - \E[\eta(\mb{x}_i)|\F_{i-1}])^2 | \F_{i-1}] \nonumber \\
         & = \textstyle \sum_{i=1}^K \frac{\Tilde{\alpha}^{2(K-i)}}{\delta^2} \textup{Var}(h(\mb{x}_i) |\F_{i-1})\leq \sum_{i=1}^K \Tilde{\alpha}^{2(K - i )} \frac{\sigma^2}{\delta^2} \label{eq:tighter_pqv}\\
         & \leq \textstyle \frac{\sigma^2 K}{\delta^2}. 
    \end{align}

    \noindent \revised{\underline{\textbf{(Step 5)}}} Now, to relate the unsafe event $\{ \min_{k \leq K} h(\mb{x}_k) < 0\} $ to our martingale $M_k$ we consider the implications:  
    \begin{align}
        \textstyle \min&_{k\leq K}   h(\mb{x}_k)  < 0  
        \implies  \textstyle \min_{k \leq K} h(\mb{x}_k)  \leq 0   \\
        \iff &  \max_{k \leq K } - \Tilde{\alpha}^{K-k} \eta(\mb{x}_k) \geq 0, \quad \textup{ since } \Tilde{\alpha} > 0, \delta > 0  \label{eq:iff_mult}\\
        \iff & \max_{k \leq K } W_k \geq \Tilde{\alpha}^K \eta(\mb{x}_0) - \textstyle \sum_{i=1}^k \Tilde{\alpha}^{K-i}\frac{\Tilde{c}}{\delta} \label{eq:iff_add_zero} \\ 
        \implies &   \max_{k \leq K } M_k \geq \Tilde{\alpha}^K \eta(\mb{x}_0)  - \textstyle \sum_{i=1}^k \Tilde{\alpha}^{K-i}\frac{\Tilde{c}}{\delta} \label{eq:mart_bound}\\
        \implies &   \max_{k \leq K } M_k \geq \Tilde{\alpha}^K \eta(\mb{x}_0)  - \textstyle \sum_{i=1}^K \Tilde{\alpha}^{K-i}\frac{\Tilde{c}}{\delta}, \label{eq:c_nonneg_bound}
    \end{align}
    where \eqref{eq:iff_mult} is due to multiplication by a value strictly less than zero, \eqref{eq:iff_add_zero} is due to adding zero, \eqref{eq:mart_bound} is due to $M_k \geq W_k$ as in \eqref{eq:doobs_mart}, and \eqref{eq:c_nonneg_bound} is due to  $k \leq K$ and the nonnegativity of $\alpha, \delta$, and $\Tilde{c}$. Thus, the unsafe event satisfies the containment:  
    \begin{align}
        \left \{ \min_{k \leq K } h(\mb{x}_k) < 0 \right\} \subseteq  \left\{ \max_{k \leq K } M_k \geq \Tilde{\alpha}^K \eta(\mb{x}_0)  - \displaystyle{\sum_{i=1}^K} \Tilde{\alpha}^{K-i}\frac{\Tilde{c}}{\delta}  \right\} \nonumber
    \end{align}
    
    Since $M_k$ satisfies $M_0=0$, $M_k - M_{k-1}\leq 1 \; \forall k \leq K$ , and $\langle M\rangle_K\leq \frac{\sigma^2K}{\delta^2}$, we can apply Thm. \ref{thm:azuma_hoeffding} (Freedman's Ineq.) with $\lambda = \frac{\Tilde{\alpha}^K h(\mb{x}_0)  - \sum_{i=1}^K \Tilde{\alpha}^{K-i}\Tilde{c}}{\delta} $ to achieve the probability bound\footnote{The proof can end \revised{after Step 5} and can be applied to any system satisfying \eqref{eq:general_expect_cond}. We specialize to DTCBFs and $c$-martingales for clarity.}: 
    $ P_u(K, \mb{x}_0) \leq H\left( \frac{\Tilde{\alpha}^K h(\mb{x}_0)  - \sum_{i=1}^K \Tilde{\alpha}^{K-i}\Tilde{c}}{\delta} , \frac{\sigma \sqrt{K}}{\delta} \right) $.

    \noindent \revised{\underline{\textbf{(Step 6)}}} If the system satisfies the \ref{eq:expect_cbf_cond} condition, then \eqref{eq:general_expect_cond} holds with $(\Tilde{\alpha} = \alpha, \Tilde{c} = 0)$ so the desired bound is achieved with $\lambda = \alpha^K h(\mb{x}_0)/\delta$ and if the system satisfies the \ref{eq:c_mart_cond} condition then \eqref{eq:general_expect_cond} holds with $(\Tilde{\alpha} = 1, \Tilde{c} = c)$ so the desired bound is achieved with $\lambda = h(\mb{x}_0)/\delta - K c$. 
\end{proof}

\ifthenelse{\boolean{long}}{
    }{
        \vspace{-1.5em}
    }

\subsection{Bound Tightness Comparison}

\begin{figure}
    \centering
    \includegraphics[width=0.9\linewidth]{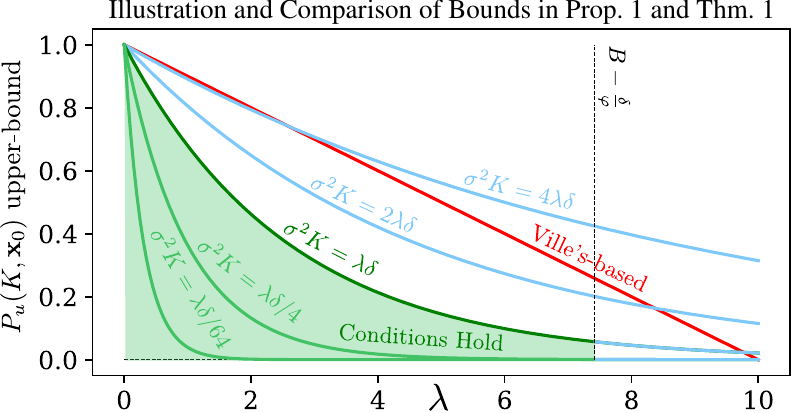}
    \caption{Comparison for Prop. \ref{thm:comparison} with $B=10, K = 100, \delta = 1,$ and varying $\sigma$ and $\lambda$. The Freedman-based bounds are shown in green when the conditions of \revised{Prop.} \ref{thm:comparison} hold and blue when they do not. The Ville's-based bound is shown in red. Code to reproduce this plot can be found at \cite{codebase}}
    \label{fig:bound_compare}
    \ifthenelse{\boolean{long}}{
        \vspace{-3em}  
    }{
        \vspace{-4em}    
    }
\end{figure}

\begin{figure*}[t!]
  \includegraphics[width=\textwidth]{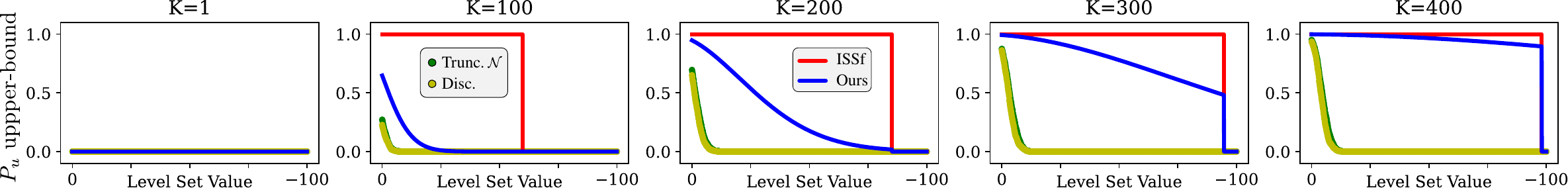}
  \caption{Probability that the system is unsafe: our bound from Cor. \ref{cor:issf} (blue), ISSf bound (red). The $x$-axis is the level set expansion $-\epsilon$ and the $y$-axis is the failure probability (lower is better). The plots from left to right indicate safety for $K = 1, 100, 200, 300, $ and $400$ steps. Simulations where $\E[h(\mb{x}_k) | \F_{k-1}] = \alpha h(\mb{x}_k) $ 
  and approximate probabilities from 1000 samples are shown for simulations where $h(\mb{x}_k)$ is sampled from 3 different conditional distributions: uniform (pink), truncated Gaussian (green), and a categorical (yellow) all which satisfy  Cor. \ref{cor:issf}. Code for these plots is can be found at \cite{codebase}. 
  }
  \label{fig:issf_compare}
  \ifthenelse{\boolean{long}}{\vspace{1em}}{    
  }{\vspace{-2em}}
\end{figure*}

We now relate the Freedman-based safety of Thm. \ref{thm:main} to the Ville's-based safety of Thm. \ref{thm:rss}. For systems that have an upper-bound $h$ \eqref{assp:upper_bound}, a lower-bounded error \eqref{assp:thm_main_bounded_diff}, and a bounded conditional variance \eqref{assp:thm_main_pqv}, we provide a range of values for $\sigma, \delta, K, B,$ and $\lambda$
for which Thm. \ref{thm:main} is stronger. \revised{This Prop. provides a direct theoretical comparison (after changing notation) to the Ville's-based bounds in \cite{steinhardt_finite-time_2012, santoyo_barrier_2021, kushner_stochastic_1967, cosner_robust_2023}.} 

 

\begin{proposition} \label{thm:comparison}

    For some $\sigma, \delta, B> 0 $, $\lambda \geq 0$ and $K \in \mathbb{N}_{1}$, consider the conditions
    \begin{align}
        \lambda\delta \geq \sigma^2 K, &&  \textstyle \lambda  \leq B - \frac{\delta}{\varphi}, \label{eq:comp_conds} 
    \end{align}
    where $\varphi = 2 \ln(2) - 1 $. If these conditions hold, then  
    \begin{align}
    \textstyle H\left( \frac{\lambda}{\delta}, \frac{\sigma\sqrt{K}}{\delta}\right) \leq 1 - \frac{\lambda}{B}.    \label{eq:comparison_thm} 
    \end{align}
\end{proposition}
\noindent 
Proof of this Proposition is is provided in the Appendix. 


Intuitively, conditions \eqref{eq:comp_conds} stipulate that 
the conditional variance $\sigma^2$ and number of steps $K$ must be limited by $\lambda\delta$, which is a function of the initial condition times the maximum single-step disturbance to $h(\mb{x}_k)$. Additionally, the initial condition must be less than the maximum safety bound $B$ by an amount proportional to $\delta$. The exact value of $\varphi$ is a result of the first assumption $(\lambda\delta \geq \sigma^2 K$) and alternative values can be found by changing this assumption; for clarity of presentation, we leave exploration of these alternative assumptions to future work. The safety bounds for various $\lambda$ and $\sigma$ are shown in Fig. \ref{fig:bound_compare} where it is clear that these conditions provide a \textit{conservative} set of parameters over which this proposition holds.


\ifthenelse{\boolean{long}}{
    }{
    }

\subsection{Extending Input-to-State Safety}

Since Thm. \ref{thm:main} assumes that $h$ has lower-bounded \revised{errors} \eqref{assp:thm_main_bounded_diff}, 
we can directly compare our method with Input-to-State Safety (ISSf) \cite{kolathaya_input--state_2019}, which provides almost-sure safety guarantees. 

In the context of our stochastic, discrete-time problem setting, the ISSf property can be reformulated as:  
\begin{proposition}[Input-to-State Safety] \label{thm:issf}
    If the closed-loop system \eqref{eq:cl_dyn} satisfies the \ref{eq:expect_cbf_cond} condition and the bounded-jump condition \eqref{assp:thm_main_bounded_diff} (a.s) for some $\alpha \in [0,1) $ and $\delta > 0 $,
    then $h(\mb{x}_{k}) \geq \alpha^k h(\mb{x}_0) - \sum_{i=0}^{k-1}\alpha^{i} \delta  $ for all $k\geq 0 $ and $\C_\delta = \left\{ \mb{x} \in \R^n  ~ |~ h(\mb{x}) \geq \frac{-\delta}{1-\alpha} \right\} $ is safe (a.s.). 
\end{proposition}
\begin{proof}
    By combining the bounds \eqref{eq:expect_cbf_cond} and \eqref{assp:thm_main_bounded_diff}: 
    \begin{align}
        h(\mb{x}_{k+1}) \geq \E[h(\mb{x}_{k+1}) ~|~ \F_{k}] - \delta \geq \alpha h(\mb{x}_k) - \delta \textup{ (a.s.)}
    \end{align}
    \noindent Thus, for all $k \in \mathbb{Z}$, we have the lower-bound $h(\mb{x}_{k}) \geq \alpha^k h(\mb{x}_0) - \sum_{i=0}^{k-1}\alpha^{i} \delta  $ (a.s). Furthermore, for all time, $h(\mb{x}_k) \geq  \frac{-\delta}{1-\alpha} \implies h(\mb{x}_{k+1})  \geq  \frac{-\delta}{1 - \alpha}  $, so 
    $\C_\delta $ is safe (a.s.). 
\end{proof}


To compare with ISSf's worst-case safe set $\C_{\delta}$, we wish to use Thm. \ref{thm:main} to bound the probability that our system leaves some expanded safe set $\mathcal{C}_\epsilon = \{ \mb{x} \in \R^n | h(\mb{x}) \geq -\epsilon\} $ with $\epsilon \geq 0 $ in finite time.
\begin{corollary} \label{cor:issf}
  If the hypotheses of Theorem \ref{thm:main} are satisfied and \eqref{eq:cl_dyn} satisfies the \ref{eq:expect_cbf_cond} for some $\alpha \in (0,1)$, then for any value $\epsilon \geq 0 $ and any $K \in \mathbb{N}_1$, 
    \begin{align}
        P & \bigg\{ \min_{k \leq K} \;  h(\mb{x}_k)     < -\epsilon \bigg\} \label{eq:p_issf_bound} \\
         & \leq \textstyle H \left(\lambda , \frac{\sigma}{\delta} \left( \frac{1 - \alpha^{2K}}{1 - \alpha^2} \right)^{\frac{1}{2}} \right) \mathds{1}_{\left\{ -\epsilon \geq \alpha^K h(\mb{x}_0) - \sum_{i=0}^{K-1}\alpha^i \delta \right\} }\nonumber 
    \end{align}
  where  $\lambda  = \frac{\alpha^K}{\delta}(h(\mb{x}_0) + \epsilon)$.
\end{corollary}
\begin{proof}
    The \ref{eq:expect_cbf_cond} condition ensures that, for any $\epsilon \geq 0$: 
    \begin{align}
        \mathbb{E}[~h(\mb{x}_{k+1}) +  \epsilon ~|~ \F_{k}~] & \geq  \alpha (h(\mb{x}_k) + \epsilon ) + \epsilon (1-\alpha) \\
        & \geq \alpha (h(\mb{x}_k) + \epsilon) 
    \end{align}
     We apply the same proof as Thm. \ref{thm:main} starting at \eqref{eq:cand_mart} with ($\eta(\mb{x}_k) = \frac{h(\mb{x}_k) + \epsilon}{\delta}, \Tilde{\alpha}= \alpha, \Tilde{c} = 0 $).  Choosing $\lambda = \alpha^K \eta(\mb{x}_0) $ and bounding\footnote{This bound on $\langle M \rangle_K$ uses the finite geometric series identity and can also be applied for a tighter Thm. \ref{thm:main} and Prop. \ref{thm:comparison}.} $\langle M \rangle_K \leq \sum_{i=1}^K \alpha^{2(K - i )} \frac{\sigma^2}{\delta^2}=\frac{\sigma^2(1-\alpha^{2K})}{\delta^2(1-\alpha^2)}$ as in \eqref{eq:tighter_pqv} yields the desired bound without the indicator function by applying Thm. \ref{thm:azuma_hoeffding}. The indicator function is a result of applying the lower bound on the safety value from Prop. \ref{cor:issf}, i.e. $h(\mb{x}_k) \geq \alpha^k h(\mb{x}_0) - \sum_{i=0}^{k-1}\alpha^i \delta $ (a.s.) for $k \in \mathbb{Z}$.  
\end{proof}



A comparison of Prop. \ref{cor:issf} and Cor. \ref{thm:issf} and Monte Carlo approximations for various $\epsilon$ and \ifthenelse{\boolean{long}}{a variety of distributions}{distributions  (truncated normal, categorical)}\footnote{Code for these simulations can be found at \cite{codebase} } is shown in Fig. \ref{fig:issf_compare} \revised{where we can see that our method successfully upper-bounds the sampled safety probabilities with risk-sensitive guarantees that are much less conservative than the worst-case bounds provided by ISSf.}

\ifthenelse{\boolean{long}}{For these simulations, we use the simple system: 
\begin{align}
    \mb{x}_{k+1} = \alpha \mb{x}_k + \mb{d}_k
\end{align}
for $\mb{x}\in \R^1$, $\alpha =0.99$, and zero-mean disturbances $\mb{d}_k$ sampled from a variety of distributions for up to $K =400$ steps. This system naturally satisfies the \ref{eq:expect_cbf_cond} constraint: 
\begin{align}
    \E[h(\mb{x}_{k+1}) | \F_k] \geq \alpha h(\mb{x}_k) \textup{ with } h(\mb{x}) = \mb{x},
\end{align}
so we seek to provide guarantees of its inherent safety probabilities. 
    In particular, in three different experiments we consider $\mb{d}_k$ sampled from one of three zero-mean distributions that all satisfy $\vert \mb{d} \vert \leq 1 $ and $\sigma \leq \frac{1}{3}$: a uniform distribution $\mathcal{U}_{[-1,1]}$, a standard normal distribution truncated at $-1$ and $1$, and a categorical distribution where $\mathbb{P}\{\mb{d} = -1  \} = \frac{1}{6} $ and $\mathbb{P}\{ \mb{d}  = \frac{1}{5} \} = \frac{5}{6}$ to ensure 0 mean. 

    These simulations show that although our method is conservative compared to the Monte-Carlo approximations, it provides useful risk-based safety probabilities for a variety of $\C_\epsilon $ level sets whereas ISSf only provides a worst-case almost-surely bound.  
}{}

\ifthenelse{\boolean{long}}{
    }{
        \vspace{-0.2em}
    }
\section{Case Study: Bipedal Obstacle Avoidance} \label{example:hlip}

In this section we apply our method to a simplified model of a bipedal walking robot. In particular, the Hybrid Linear Inverted Pendulum (HLIP) model \cite{xiong20223} approximates a bipedal robot as an inverted pendulum with a fixed center of mass (COM) height $z_0 \in \R_{>0}$. Its states are the planar position, relative COM-to-stance foot position, and COM velocity $\mb{p},  \mb{c}, \mb{v} \in \R^2 $. The step-to-step dynamics are linear and the input is the relative foot placement, $\mb{u}_k \in \R^2$. The matrices $\mb{A} \in \R^{6\times 6}$ and $\mb{B}\in \R^{6\times2} $ are determined by $z_0$ and gait parameters including the stance and swing phase periods. The HLIP model with \revised{an added disturbance matrix $\mathbf{D}\in \R^{6 \times 4}$ and disturbance $\mb{d}\in \R^4$ affecting position and velocity} is: 
\begin{align}
    \mb{x}_{k+1}= \mb{A} \mb{x}_{k} + \mb{B}\mb{u}_k + \mb{D} \mb{d}_k, \quad \mb{d}_k \sim \mathcal{D}. \nonumber
\end{align}
where $\mb{x}_{k} = \lmat \mb{p}_{k}^\top& \mb{c}_{k}^\top &
    \mb{v}_{k}^\top
    \rmat^\top$. We augment the standard HLIP model and assume that $\mb{d}$ enters linearly and $\cal{D}$ is a $4$-dimensional, $0$-mean uniform distribution\footnote{\ifthenelse{\boolean{long}}{
            See Appx. \ref{appx:construct_delta_example}
        }{
            See \cite[Appx. H]{cosner2024bounding} 
        }
        for bounds for $\delta$ and $\sigma$ given this problem structure. 
        } with $\Vert \mb{d} \Vert \leq d_\textup{max}$. 

 We define safety for this system as avoiding a circular obstacle of radius $r> 0$ located at $(x,y) = \bs{\rho} \in \R^2$, so safety can be defined using the signed-distance function $
     h(\mb{x}) = \Vert \mb{p} - \bs{\rho} \Vert_2 - r$. Notably, this function has no upper bound and therefore the Ville's-based Thm. \ref{thm:rss} does not apply. 

Since $h(\mb{x})$ is not convex, we use a conservative halfspace convexification instead: 
\begin{align}
    h(\mb{x}_{k+1}) \geq\widehat{\mb{e}}(\mb{p}_k) ^\top\left( \mb{p}_{k+1}   - \bs{\rho} \right) - r \triangleq \bar{h}(\mb{x}_{k+1}), 
\end{align}
\noindent where $\widehat{\mb{e}}(\mb{p}) = \frac{(\mb{p}- \bs{\rho})}{\Vert \mb{p}  - \bs{\rho} \Vert }$ and we apply the controller: 
\begin{align}
    \mb{u}^* = \min_{\mb{u} \in \R^2} & \quad \Vert \mb{u} - \mb{k}_\textup{nom}(\mb{x}_k) \Vert\\
    \textup{s.t. } & \quad \mathbb{E}\left[~\bar{h}(\mb{x}_{k+1})  ~|~ \F_k~\right] \geq \alpha \bar{h}(\mb{x}_k) \nonumber 
\end{align}
\noindent with $\alpha \in (0,1]$ and  where $\mb{k}_\textup{nom}$ tracks a desired velocity. 

We ran 5000 trials with 3 steps per second and compared against the theoretical bound from Thm. \ref{thm:main}. Those values and planar pose trajectories can be seen in Fig. \ref{fig:hlip}. Exact values and code for this and all other plots can be found in \cite{codebase}.

\ifthenelse{\boolean{long}}{
    }{
        \vspace{-1.0em}
    }

\section{Conclusion}

Despite the relative tightness guarantee of Prop. \ref{thm:comparison}, the probability guarantees of our method are not necessarily tight, as can be seen in Fig. \ref{fig:issf_compare}. Optimization of $h$ without changing $\C$ as in \cite{steinhardt_finite-time_2012} is a promising direction further tightening. 
Additionally, the case study shown in Section \ref{example:hlip} presents an immediate direction for future work which may involving a hardware demonstration of this method.

\ifthenelse{\boolean{long}}{
    }{
        \vspace{-0.5em}
    }

\appendix 

\ifthenelse{\boolean{long}}{
    \subsection{Proof of Ville's Inequality}    \label{appdx:villes}
    
\begin{proof}
    Fix $\lambda>0$ and define the stopping time $\tau \triangleq \inf \{k\in \mathbb{N} ~|~  \mart_k > \lambda  \} $ with $\tau = +\infty$ if $\mart_k \leq \lambda $ for all time. Since $\mart_k$ is a nonnegative supermartingale, the stopped process $\mart_{k \wedge \tau}$ is also a nonnegative supermartingale where
    \begin{align}
        \E[\mart_{k \wedge \tau }] \leq \E[\mart_0] \textup{ and } \liminf_{k \to \infty } \E [\mart_{k \wedge \tau}] \leq \E[\mart_0] .
    \end{align}
    We can further bound this in the case that $\tau $ is finite: 
    \begin{align}
        \E[\mart_0] & \geq \liminf_{k \to \infty} \E[\mart_{k \wedge \tau} \mathds{1}_{\{\tau < \infty \} }]\\
        & \geq \E [\liminf_{k \to \infty} \mart_{k \wedge \tau} \mathds{1}_{\{\tau < \infty \} }  ]\\
        & > \E[\lambda \mathds{1}_{\tau < \infty } ] = \lambda \P\{\tau < \infty  \} = \lambda \P\left\{\sup_{k\in \mathbb{N} } W_k > \lambda \right\} \nonumber. 
    \end{align}
    The first inequality is by the nonegativity of $W_k$, the second inequality is by Fatou's Lemma \cite{grimmett_probability_2020}, and the third is by the definition of $\tau$. 
    Rearranging terms completes the proof.
\end{proof}

    \subsection{Proof of Theorem \ref{thm:rss}} \label{appdx:rss_and_stein}
    \begin{proof}
    We prove the two cases separately:
    
    \begin{itemize}
        \item We first prove the case when \eqref{eq:expect_cbf_cond} is satisfied.
    Let $\mart_k \triangleq B\alpha^{-K} -\alpha^{-k} h(\mb{x}_k) $. This is a nonnegative supermartingale for $k \leq K $: 
    \begin{align}
        \mart_k   = \alpha^{-K}B &  - \alpha^{-k}h(\mb{x}_k) \geq \alpha^{-k}(B - h(\mb{x}_k)) \geq 0  \nonumber \\ 
        \E[\mart_{k+1}|\F_k] &= \alpha^{-K}B - \alpha^{-(k+1)}\E[h(\mb{x}_{k+1})|\F_{k}] \nonumber   \\
        & \leq \alpha^{-K}B - \alpha^{-k}h(\mb{x}_k) = \mart_k. 
    \end{align}
    Apply Ville's inequality \eqref{thm:villes} to $\mart_k$ to find: 
    \begin{align}
        \P\left\{ \max_{k \leq K } \mart_k > \lambda  \right\} \leq \frac{\E[\mart_0]}{\lambda}.  
    \end{align}
    Next note that the implication: 
    \begin{align}
        \exists k \leq K \textup{ s.t. } h(\mb{x}_k) < 0 \implies \exists k \leq K \textup{ s.t. } \mart_k > \alpha^{-K}B \nonumber 
    \end{align}
    ensures that $P_u(K, \mb{x}_0)  \leq \mathbb{P} \left\{ \max_{k \leq K } \mart_k > \alpha^{-K} \right\} $. Choose $\lambda = \alpha^{-K} B$ to achieve: 
    \begin{align}
        P_u(K, \mb{x}_0) \leq \frac{ \alpha^{-K}B- h(\mb{x}_0) }{\alpha^{-K}B} = 1 - \frac{h(\mb{x}_0)}{B}\alpha^K
    \end{align}

        \item     Next we prove the case when \eqref{eq:c_mart_cond} is satisfied.
            Let $\mart^c_k \triangleq B-h(\mb{x}_k) + (K - k) c $. This is a non-negative supermartingale for $k \leq K$: 
            \begin{align}
            \mart^c_k &  = B - h(\mb{x}_k)  + (K - k)c \geq 0  \\ 
                \E[ \mart^c_{k+1} ~|~ \F_k] & = B-\E[h(\mb{x}_{k+1}) ~|~ \F_{k}] + (K - k-1)c \nonumber \\
                & \leq B-h(\mb{x}_k) + c +(K - k -1) c \\
                & = B -h(\mb{x}_k) + (K - k) c = \mart^c_k
            \end{align}
            Apply Ville's inequality \eqref{thm:villes} to $\mart^c_k$ to find: 
            \begin{align}
                \P\left\{ \max_{k\leq K } \mart^c_k > \lambda  \right\} \leq \frac{\E [\mart^c_0]}{\lambda } 
            \end{align}
            Next note that the implication: 
            \begin{align}
                \exists k \leq K \textup{ s.t. } h(\mb{x}_k) < 0 \implies \exists k \leq K \textup{ s.t. } \mart^c_k > B  \nonumber 
            \end{align}
            ensure that $P_u(K, \mb{x}_0)  \leq \mathbb{P}\{ \max_{k \leq K } \mart^c_k > \lambda \}$. 
            Choose $\lambda = M $ to achieve: 
            \begin{align}
                P_u(K, \mb{x}_0) \leq \frac{B - h(\mb{x}_0) + K c}{B } = 1 - \frac{h(\mb{x}_0) - Kc}{B}. \nonumber
            \end{align}

    \end{itemize}
       
    \end{proof}

    \subsection{\revised{Ville's-based Safety Theorems from  \cite{santoyo_barrier_2021} and \cite{steinhardt_finite-time_2012}}} \label{appx:santoro}
    \revised{To show how Prop. \ref{thm:comparison} can be used to compare with existing literature, we restate \cite[Thm. 2]{santoyo_barrier_2021} which contains \cite[Thm, 2.3]{steinhardt_finite-time_2012}. In particular, using the transformation, $h(\mb{x}) = B( 1 - b_s(\mb{x}))$ where $b_s(\mb{x})$ is the relevant safety function from \cite{santoyo_barrier_2021}, \cite[Thm. 2]{santoyo_barrier_2021} can be rewritten as: }
    \revised{
    \begin{theorem}[{\cite[Thm. 2]{santoyo_barrier_2021}}]
        Given the closed-loop dynamics \eqref{eq:cl_dyn}
        and the sets $\mathcal{X} \subset \R^n$, $\mathcal{X}_0 \subset \textup{Int}(\mathcal{C})$. Suppose there exists a twice differentiable function $h$ such that
        \begin{align}
            h(\mb{x}) & \leq B, \quad \forall \mb{x} \in \mathcal{X},  \\
            \E [h(\mb{F}(\mb{x}_k)) ~|~ \mb{x}_k] & \leq \alpha h(\mb{x}_k) - c, \quad \forall \mb{x} \in \mathcal{C}
        \end{align}
        for some $\alpha \in (0,1] , c \in [B(\alpha-1), \alpha) $, $\gamma \in [B, 0)$.
        Then,
        \begin{align}
            & \textup{if $\alpha \in (0,1]$ and $c \leq 0$, } \textstyle P_u(\mb{x}_0) \leq 1 - \frac{h(\mb{x}_0)}{B} \left( \frac{\alpha B - c}{ B }\right)^{K} \label{eq:santoyo_safe_disturbance} \\
            & \textup{if $\alpha \in (0,1]$ and $c > 0$, } \nonumber \\
            & \quad \quad\quad \quad \quad\quad  \textstyle P_u(\mb{x}_0, K) \leq 1 - \frac{\alpha^K h(\mb{x}_0)}{(1 - \alpha)} \left(\frac{c + (1 - \alpha) B}{B} \right) \label{eq:santoyo_rss}\\
            & \textup{if $\alpha \in (0,1]$ and $c = 0$, } \displaystyle P_u(\mb{x}_0, K) \leq 1 - \frac{\alpha^K h(\mb{x}_0)}{B} \label{eq:santoyo_dtcbf}\\
            &  \textup{if $\alpha = 1$, }\displaystyle P_u(\mb{x}_0, K) \leq 1 - \frac{h(\mb{x}_0) - cK}{B} \label{eq:santoyo_c_mart}.
        \end{align}
    \end{theorem}
    One minor difference between these methods is that we consider a state $\mb{x}$ to be safe if $\mathbf{x}\in \mathcal{C}$ and in \cite{santoyo_barrier_2021} they consider a state to be safe if it is in $\textup{Int}(\mathcal{C})$. Ultimately, this makes very little difference for the bound because $\mb{x} \notin \mathcal{C} \implies \mb{x} \notin \textup{Int}(\mathcal{C})$ and because $\mathcal{C} \setminus \textup{Int}(\mathcal{C})$ is generally a set of zero measure. 
    }



    \subsection{\revised{Proof of Freedman's Inequality, (Theorem \ref{thm:azuma_hoeffding})}}\label{appdx:freedman}
    
 \revised{ Here we represent the proof as given in \cite{freedman1975tail} for reference. }

    \revised{First, Freedman presents the following bound on the exponential function. We refer to \cite[Lem. 3.1 and Cor. 3.2]{freeman_global_1995} for the original calculus. }

    \revised{
    \begin{lemma}[{\cite[{Cor. 3.2}]{freedman1975tail}}]\label{lemma:freedman} For $\varphi(\gamma) \triangleq (e^\gamma -1 - \gamma) $
        \begin{align}
            e^{\gamma x} \leq 1 + \gamma x + x^2 \varphi(\gamma) , \textup{ for all } \gamma \geq 0, x \leq 1.  
        \end{align} 
    \end{lemma}
    }
    
    \revised{
    Freedman then uses this construction to bound the moment generating function of supermartingale differences to find the following bound from which the main theorem will follow almost immediately. 
    \begin{proposition}[{\cite[{Prop. 3.3}]{freedman1975tail}}]
        If $W_k = \sum_{i=0}^k X_k $ is a supermartingale with stopping time $\tau$ where  $ W_k - W_{k-1} \leq 1 \quad (\textup{a.e.})$ 
    for all $k \in \mathbb{N}_1$ and any $\gamma \geq 0 $, then 
    \begin{align}
        \int_{\{\tau < \infty\}} e^{\gamma W_\tau   - \varphi(\gamma)\langle W \rangle_\tau } d \mathbb{P} \leq 1 . \label{eq:prop_wts}
    \end{align}
    \end{proposition}
    }
    \begin{proof}
    \revised{
        Let $p_x(x) $ be the probability distribution of a random variable $X$ with $X \leq 1$ (a.e.) and $\E[X] \leq 0 $ (a.e.). Choose probability distributions $p_0(x)$ on $(-\infty , 1]$ and $p_-(x)$ on $(-\infty, 0)$ such that 
        \begin{align}
         \mathbb{E}_0[X] &= \int_{(-\infty, \infty)} x p_0(x) dx = 0, \\
         p_x(x) &= \theta p_0(x) + (1 - \theta) p_-(x). 
        \end{align}
        Next we bound the moment generating function $\E\left[e^{\gamma X}\right]$, 
        \begin{align}
            & \E\left[e^{\gamma X}\right] =  \int_{(-\infty, 1]} e^{\gamma x} (\theta p_0(x) + (1 - \theta) p_-(x))dx \\
            & = \theta \E_0[e^{\gamma X}] + (1 - \theta) \int_{(-\infty, 0)}e^{\gamma x}p_-(x) dx   \\
            & \leq \theta \int_{(-\infty, 1] }e^{\gamma x} p_0(x) dx + (1- \theta)\label{eq:exp}\\
            & \leq \theta \int_{(\infty, 1]} (1 + \gamma x + x^2\varphi(\gamma)) p_0(x) dx + (1 - \theta) \label{eq:use_prop}\\
            & = \theta (1 + 0 + \E_0[X^2]\varphi(\gamma))  + (1 - \theta)\\
            & \leq \theta (1 + \E_0[X^2]\varphi(\gamma))  + (1 - \theta)(1 + \varphi(\gamma) \textup{Var}_-(X)) \label{eq:use_phi}\\\
            & \leq 1 + \varphi(\gamma)\big(\E[X^2] - (1 - \theta)^2\E_-[X]^2\big) \label{eq:use_theta}\\
            & =   1 + \varphi(\gamma)\big(\E[X^2] - \E[X]^2\big)  = 1 + \varphi(\gamma) \textup{Var}(X),  \label{eq:mean_def}\\
            & \leq e^{\varphi(\gamma) \textup{Var}(X)} \implies \E[e^{\gamma X - \varphi(\gamma) \textup{Var}(X)}] \leq 1  \label{eq:linear}
        \end{align}
        where \eqref{eq:exp} is attained by bounding $e^{\gamma x} \leq 1$ since $\gamma \geq 0 $ and $x\in(-\infty, 0)$ and then using the fact that $p_-$ is a probability distribution, \eqref{eq:use_prop} is attained by using of Lem. \ref{lemma:freedman}, \eqref{eq:use_phi} is attained by noting that $\varphi(\gamma)\geq 0 $ and $ \textup{Var}_-(X) \geq 0 $, \eqref{eq:use_theta} is attained by noting that $\E[X^2] = \theta \E_0[X^2] + (1 - \theta)\E_-[X^2]$ and that $(1- \theta) \in [0,1]$ and $\E_-[X]^2 \geq 0 $, \eqref{eq:mean_def} holds with equality since $\E_0[X] = 0$, and \eqref{eq:linear} holds due to the bound $ 1 + x \leq e^x$ for $ x\geq 0$. 
        }

        \revised{
        This allows us to establish that $Q_k \triangleq e^{\gamma W_k   - \varphi(\gamma)\langle W \rangle_k }$ is a supermartingale since:
        \begin{align}
            &\E[Q_{k+1}~|~\mathscr{F}_{k}] = \E[e^{\gamma W_{k+1}   - \varphi(\gamma)\langle W \rangle_{k+1} } ~|~ \mathscr{F}_{k}]\\
            & = Q_k\E\left[e^{\gamma(W_{k+1} - W_k)   - \varphi(\gamma)(\langle W \rangle_{k+1} - \langle W \rangle_{k}) }~|~ \mathscr{F}_k\right] \leq Q_k \nonumber 
        \end{align}
        which holds since  $X \triangleq W_{k+1} - W_{k}$ given $\mathscr{F}_k$ satisfies \ref{eq:linear}. 
        }

        \revised{
        Next we note that $Q_0 = 1 $ and $Q_{k\land \tau}$ is also a positive supermartingale, so 
        \begin{align}
            1 & \geq \liminf_{\tau \to \infty } \E[Q_{k \land \tau   }] \geq \liminf_{\tau \to \infty } \E[Q_{k \land \tau}  \mathds{1}_{\{\tau < \infty\}}]\\
            & \geq  \E[\liminf_{\tau \to \infty } Q_{k \land \tau} ] = \E[Q_{\tau} \mathds{1}_{\{\tau < \infty\}}],
        \end{align}
        where (as in \cite[{Proof of Thm. 2.3}]{tropp2011freedman}
        ) the indicator decreases the expectation because $Q_{k\land \tau}$ is positive, Fatou's lemma \cite{grimmett_probability_2020} justifies the third inequality, and the fact that $\tau < \infty $ for the indicator event yields the final equality which is equivalent to \eqref{eq:prop_wts} as desired.
        }
    \end{proof}

    \revised{
    \begin{theorem*}
        [{Freedman's Inequality \cite[\revised{Thm. 4.1}]{freedman1975tail}}]
    If, for some $K \in \mathbb{N}_1$ and $\xi > 0 $, $\mart_k$ is a supermartingale with $W_0=0$ such that: 
    \begin{align}
        (\mart_k - \mart_{k-1})\leq 1&  \quad \textup{ for all } k \leq K,   \tag{Assumption \ref{assp:bounded_mart_diff}}\\
        \langle \mart\rangle_K \leq \xi^2,&  \tag{Assumption \ref{assp:bounded_pred_quad_var}}
    \end{align}
     then, for any $\lambda \geq 0 $, 
    \begin{align}
        \P \left\{ \max_{k \leq K } \mart_k \geq \lambda  \right\} \leq H(\lambda, \xi) \triangleq \left( \frac{\xi^2}{\lambda + \xi^2}\right)^{\lambda + \xi^2} e^{\lambda }.   \label{eq:apx_freedman_wts}
    \end{align}
    \end{theorem*}
    }

    \begin{proof}
        \revised{
        Define the stopping time $\tau$ as the smallest $k\leq K$ such that $W_k \geq \lambda$, and $\tau = \infty $ if $W_k < \lambda$ for all $k \leq K$. Also define the event $A \triangleq \{ W_k \geq \lambda \textup{ and } \tau < \infty \textup{ for some } k \leq K \} $. 
        }

        \revised{
        Next, we continue by bounding using any $\gamma$: 
        \begin{align}
            1 &\geq \int_{A} \exp\{\gamma W_\tau - (e^\gamma -1 - \gamma)\langle W_k\rangle \} d \P \\
            & \geq \int_A \exp\{\gamma \lambda  - (e^\gamma -1 -\gamma) \xi^2\}d \P\\
            & = \P\{A\}   \exp\{\gamma \lambda  - (e^\gamma -1 -\gamma) \xi^2\}\\
            \implies & \P\{A\} \leq \exp\{(e^\gamma -1  - \gamma) \xi^2 - \gamma \lambda\}
        \end{align}
        From here we choose $\gamma= \ln\left( \frac{\lambda + \xi^2}{\xi^2}\right)$ to minimize this probability bound and achieve the desired result \eqref{eq:apx_freedman_wts}.  
        }
    \end{proof}

    \subsection{Proof of Proposition \ref{thm:comparison}} \label{appdx:comparison}
    
\begin{proof}
    Define $\Delta\left(\lambda, B, \sigma, K, \delta \right) \triangleq 1 - \frac{\lambda }{B } - H\left(\frac{\lambda}{\delta}, \frac{\sigma \sqrt{K}}{\delta}\right).  \label{eq:comparison_delta_fun}$ 

    If $\Delta(\lambda, B, \sigma, K, \delta) \geq 0$, then \eqref{eq:comparison_thm} must hold. We first show 
    $\Delta$ is monotonically decreasing in $\sigma^2$. Consider\ifthenelse{\boolean{long}}{\footnote{The derivation of this derivative is given after the proof.}}{\footnote{The derivation of $\frac{\partial \Delta }{\partial (\sigma^2)}$ is provided in our extended manuscript \cite{cosner2024bounding}.}}
     $\frac{\partial \Delta }{\partial (\sigma^2)} = a(\lambda, \sigma, K, \delta) b(\lambda, \sigma, K, \delta)$  
    where
    \begin{align}
        \textstyle a(\lambda, \sigma, K, \delta)  & \textstyle \triangleq  
        \frac{-e^\frac{\lambda}{\delta}}{\delta^2\sigma^2}\left( \frac{\sigma^2K}{\lambda\delta+ \sigma^2K}\right)^{\frac{( \lambda\delta + \sigma^2K )}{\delta^2}} <0, \\
        \textstyle b(\lambda, \sigma, K, \delta) & \textstyle \triangleq \left( \sigma^2K \ln \left( \frac{\sigma^2K}{\lambda \delta + \sigma^2K }\right) + \lambda\delta  \right).
    \end{align}
    The function $a(\lambda, \sigma, K, \delta)$ is negative since $\delta, \sigma, K > 0 $. For $b(\cdot)$, the logarithm bound $\ln(r) \geq 1 - 1/r $ ensures that: 
    \begin{align}
        \textstyle b(\lambda,\sigma, K, \delta) \geq \sigma^2K \left( 1 - \frac{\lambda\delta + \sigma^2K}{\sigma^2K}\right) + \lambda\delta = 0.  
    \end{align}
    Since $a<0$ and $b\geq 0 $,  $\Delta(\lambda, B, \sigma, K, \delta) $ is monotonically decreasing with respect to $\sigma^2$, so we can use the assumption $\sigma^2 K  \leq \lambda\delta$ to lower bound $\Delta$ as: 
    \begin{align}
        &\Delta(\lambda, B, \sigma, K, \delta)  \geq 
        \textstyle 1 - \frac{\lambda}{B} - \left( \frac{1}{2}\right)^{2\frac{\lambda}{\delta}}e^{\frac{\lambda}{\delta}}\\ 
        & =\textstyle  1 - \frac{\lambda}{B} - e^{(1 - 2\ln(2))\frac{\lambda}{\delta}} \triangleq 1 - \frac{\lambda}{B} - e^{-\varphi \frac{\lambda}{\delta} } \triangleq \Delta_1(\lambda, B, \delta)\nonumber 
    \end{align}
    where $\varphi  \triangleq 2 \ln(2) - 1 >0$. 
    
    Next, we show that $\Delta_1(\lambda, B, \delta)\geq 0 $ for\footnote{This interval is non-empty since $\lambda \geq 0 $ and $B\geq \lambda + \frac{\delta}{\varphi}$ implies $B \geq \frac{\delta}{\varphi}$.} $\lambda \in \left[0,B-\frac{\delta}{\varphi}\right]$. 
    We prove this by showing that $\Delta_1(\lambda, B, \delta) \geq 0 $ for $\lambda = \left\{ 0, B - \frac{\delta}{\varphi}\right\} $ and that $\Delta_1$ is concave with respect to $\lambda$. 

    \vspace{1em}
    \noindent \underline{\textbf{(1)} Nonnegativity at $\lambda = 0 $:} \; $\Delta_1(0, B, \delta)   = 0.  \quad \quad \quad \quad \quad \quad  $
    \noindent\underline{\textbf{(2)}   Nonnegativity at $\lambda = B-\frac{\delta}{\varphi}$:} 
    \begin{align}
          \textstyle\Delta_1&\textstyle\left( B -\frac{\delta}{\varphi}, B , \delta \right)  = \frac{\delta}{\varphi B}- e^{-\left(B-\frac{\delta}{\varphi}\right)\frac{\varphi}{\delta}} = \frac{\delta}{\varphi B}- e^{\left(1 - \frac{B\varphi}{\delta}\right)} \nonumber\\
          & \textstyle\geq \frac{\delta}{\varphi B} - \frac{\delta}{\varphi B} = 0, \label{ineq:exp}
    \end{align} 
    where the inequality in line \eqref{ineq:exp} is due to the previously used log inequality: $\ln(r) \geq 1 - \frac{1}{r} \iff r \geq e^{\left(1 - \frac{1}{r}\right) }$, which holds for $r = \frac{\delta}{B \varphi}>0$ since $B, \delta, \varphi > 0 $. 
    \vspace{0.5em}
    
    \noindent \underline{\textbf{(3)} Concavity for $\lambda \in [0, B-\textstyle\frac{\delta}{\varphi}]$:}
    Since $\frac{\varphi}{\delta} > 0$, the second derivative of $\Delta_1$ with respect to $\lambda$ is negative:  
    \begin{align}
    \textstyle \frac{\partial^2  \Delta_1}{\partial \lambda^2}  = -\left(\frac{\varphi}{\delta}\right)^2 e^{- \varphi\frac{\lambda}{\delta}} < 0.
    \end{align}
    Thus, $\Delta_1$ is concave with respect to $\lambda.$ 
    Since, $\Delta_1(0, B, \delta)\geq 0$, $\Delta_1\left(B-\frac{\delta}{\varphi}, B, \delta  \right) \geq 0 $, and $\Delta_1(\lambda , B, \delta)$ is concave for all $\delta > 0 $ and $B \geq \frac{\delta}{\varphi}$, it follows from the definition of concavity that $\Delta_1(\lambda , B, \delta) \geq 0 $ for all $\lambda \in \left[0, B-\frac{\delta}{\varphi}\right].$

    Using this lower bound for $\Delta_1(\lambda, B) $, we have
    $\Delta(\lambda, B, \sigma, K, \delta ) \geq \Delta_1(\lambda, B) \geq 0 $
    which implies the desired inequality \eqref{eq:comparison_thm}.
\end{proof}

    \subsection{Derivative of $\frac{\partial \Delta}{\partial (\sigma^2)}$} \ref{appdx:comparison}
Here we show the derivation of the derivative given in \eqref{eq:comparison_delta_fun}. For reference, the complete function is: 
\begin{align}
    & \Delta(\lambda, B, \sigma, K, \delta) \triangleq 1 - \frac{\lambda}{B} - \left( \frac{\sigma^2 K}{\lambda \delta + \sigma^2K}\right)^{\frac{1}{\delta^2}(\sigma^2K + \lambda \delta)}e^{\frac{\lambda}{\delta}} \nonumber 
\end{align}
with the partial derivative with respect to $\sigma^2$: 
\begin{align}
     \frac{\partial \Delta}{\partial (\sigma^2) } 
    & =   - e^{\frac{\lambda}{\delta^2}} \frac{\partial}{\partial (\sigma^2)}\left[ u(\sigma^2)^{v(\sigma)}\right]\\
    & =  - e^{\frac{\lambda}{\delta^2}}\frac{ u(\sigma^2)^{v(\sigma^2)} }{u(\sigma^2)^{v(\sigma^2)}} \frac{\partial}{\partial (\sigma^2)}\left[ u(\sigma^2)^{v(\sigma^2)}\right] \label{eq:times_1}\\
    & = - e^{\frac{\lambda}{\delta^2}} u(\sigma^2)^{v(\sigma^2)} \frac{\partial}{\partial (\sigma^2)}\left[ \ln\left(u(\sigma^2)^{v(\sigma^2)}\right)\right] \label{eq:reverse_prodrule}\\
    & =  - e^{\frac{\lambda}{\delta^2}} u(\sigma^2)^{v(\sigma^2)} \frac{\partial}{\partial (\sigma^2)}\left[ \ln\left(u(\sigma^2)\right) v(\sigma^2)\right]\label{eq:log_prop}\\
    & = - e^{\frac{\lambda}{\delta^2}} u(\sigma^2)^{v(\sigma^2)} \left[ \frac{v(\sigma^2)}{u(\sigma^2)}\frac{\partial u }{\partial (\sigma^2)} + \ln(u(\sigma^2)) \frac{\partial v}{\partial (\sigma^2)} \right]  \nonumber \\
    & =  - e^{\frac{\lambda}{\delta^2}} u(\sigma^2)^{v(\sigma^2)} \left[ \frac{\lambda}{\delta \sigma^2  } + \ln(u(\sigma^2)) \frac{K}{\delta^2} \right] \nonumber\\
    & = \underbrace{ - \frac{e^{\frac{\lambda}{\delta^2}}}{\delta^2 \sigma^2 } u(\sigma^2)^{v(\sigma^2)}}_{\triangleq a(\lambda, \sigma, K, \delta)} \underbrace{\left[ \lambda \delta  +  \sigma^2K \ln(u(\sigma^2)) \right]}_{\triangleq b(\lambda, \sigma, K, \delta)}
\end{align}
where introduce the following functions for clarity:
\begin{align}
    u(\sigma^2) \triangleq \frac{\sigma^2K}{\delta \lambda + \sigma^2 K }, && 
    v(\sigma^2) \triangleq \frac{1}{\delta^2}(\lambda \delta + \sigma^2 K).
\end{align}
Critically, this proof multiplies by 1 in line \eqref{eq:times_1} (which is well defined since $\sigma, \delta, K > 0$), then applies the product rule in reverse \eqref{eq:reverse_prodrule}, and then uses the properties of the logarithm function \eqref{eq:log_prop}. The derivation is finished by applying the product rule and rearranging terms.

    \subsection{Sufficient Conditions for Constructively Bounding $\delta$ and $\sigma$}\label{appx:construct_delta}
    Here we provide sufficient conditions for which bounds on $\delta$ and $\sigma $ in Theorem \ref{thm:main}'s assumptions \eqref{assp:thm_main_bounded_diff} and \eqref{assp:thm_main_pqv} are constructive. 
    \revised{
        \begin{proposition}
            If 
            \begin{align}
                & \mb{d}\sim \mathcal{D} \textup{ satisfies } \Vert\mb{d}\Vert\leq d_\textup{max} \textup{ for some } d_{\max} \geq 0, \\
                & h:\R^n\to \R \textup{ is globally Lipschitz with } \mathcal{L}_h\geq 0, \\
                & \mb{x}_{k+1} = \mb{F}(\mb{x}_k) + \mb{d}_k, \textup{ for some $\mb{F}: \R^n \to \R^n,$}
            \end{align}   
            then 
            \begin{align}
                \E[h(\mb{x}_k) ~|~ \mathscr{F}_{k-1} ] - h(\mb{x}_k) & \leq 2\mathcal{L}_hd_{\max}   \triangleq \delta, \\
                \textup{Var}(h(\mb{x}_{k+1})~|~ \mathscr{F}_k) & \leq  \mathcal{L}_h^2 d_{\max}^2 \triangleq \sigma^2
            \end{align}
        \end{proposition}
        \begin{proof}            
            First we bound $\delta$:
            \begin{align}
                &\E[h(\mb{x}_k) ~|~ \mathscr{F}_{k-1}] - h(\mb{x}_k) \\
                & = \mathbb{E}[h(\mb{F}(\mb{x}_{k-1}) + \mb{d}_{k-1}) ~|~\mathscr{F}_{k-1}] - h(\mb{F}(\mb{x}_{k-1}) + \mb{d}_{k-1})  \nonumber \\
                & \leq \mathbb{E}[h(\mb{F}(\mb{x}_{k-1})) + \mathcal{L}_h \Vert \mb{d}_{k-1}\Vert ~|~ \mathscr{F}_{k-1} ] \\
                & \quad \quad \quad \quad \quad \quad \quad \quad \quad \quad \quad \quad - h(\mb{F}(\mb{x}_{k-1})) + \mathcal{L}_h \Vert \mb{d}_{k-1}\Vert \nonumber   \\
                &  = h(\mb{F}(\mb{x}_{k-1})) - h(\mb{F}(\mb{x}_{k-1}))+  \mathcal{L}_h\mathbb{E}[ \Vert \mb{d}_{k-1}\Vert ]  + \mathcal{L}_h \Vert \mb{d}_{k-1}\Vert \nonumber \\
                & \leq \mathcal{L}_h\mathbb{E}[d_{\max} ] + \mathcal{L}_h d_{\max} = 2\mathcal{L}_hd_{\max}\triangleq \delta
            \end{align}
            To bound $\sigma^2$, note that boundedness of $\mathcal{D}$ and  Lipschitz continuity of $h$ implies that: 
            \begin{align}
                h(\mb{F}(\mb{x}_k)) - \mathcal{L}_h d_{\max} \leq h(\mb{F}(\mb{x}_k) + \mb{d}_k) \leq h(\mb{F}(\mb{x}_k)) + \mathcal{L}d_{\max}. \nonumber 
            \end{align}
            Thus, the distribution of $h(\mb{F}(\mb{x}_k) + \mb{d}_k)$ is bounded at $\mathscr{F}_k$, so we can use Popoviciu's inequality on variances \cite{popoviciu1935equations} to achieve: 
            \begin{align}
                \textup{Var}(h(\mb{F}(\mb{x}_k) + \mb{d} ~|~ \mathscr{F}_{k-1}) \leq \mathcal{L}_h^2d_{\max}^2 \triangleq \sigma^2
            \end{align}
        \end{proof}
    }

    \subsection{Bounding $\delta$ and $\sigma^2$ in the Example}\label{appx:construct_delta_example}

 \revised{   
 The bound (14) can be obtained for the example in Section \ref{example:hlip} by using the given assumptions that $\mathcal{D}$ is uniform on the ball of radius $d_\textup{max}$ and that the matrices $\mb{C}$ and $\mb{D}$ reflect the fact that safety is defined only with respect to position and that the global position and the center-of-mass (COM) position are coupled. These facts give $\mb{C}$ and $\mb{D}$ this structure: 
\begin{align}
    \mb{C} = \lmat 1& 0 & 0 & 0& 0 & 0 \\ 0 & 1 & 0 & 0 & 0 & 0\\ 0 & 0 & 0 & 0& 0 & 0\\ 0 & 0 & 0 & 0& 0 & 0\\ 0 & 0 & 0 & 0& 0 & 0\\ 0 & 0 & 0 & 0 & 0 & 0 \rmat, &&  \mb{D} = \lmat 1 & 0 & 0 & 0 \\
        0 & 1 & 0 & 0 \\
        1 & 0 & 0 & 0 \\
        0 & 1 & 0 & 0 \\
        0 & 0 & 1 & 0 \\
        0 & 0 & 0 & 1\rmat 
\end{align}
In this case we can calculate $\delta$ as: 
\begin{align}
    & \mathbb{E}[h(\mb{x}_k) ~|~ \mathscr{F}_{k-1}] - h(\mb{x}_k) \\
    & =     \mathbb{E}[h(\mb{C}(\mb{A}\mb{x}_{k-1} + \mb{B}\mb{u}_{k-1} + \mb{D}\mb{d}_{k-1})] ~|~ \mathscr{F}_{k-1}]\\
    & \quad\quad \quad \quad\quad \quad \quad\quad \quad  - h(\mb{C}(\mb{A}\mb{x}_{k-1} + \mb{B}\mb{u}_{k-1} + \mb{D}\mb{d}_{k-1}))\nonumber \\
    & = \mathbb{E}[\Vert \mb{C}(\mb{A}\mb{x}_{k-1} + \mb{B}\mb{u}_{k-1} + \mb{D}\mb{d}_{k-1} - \bs{\rho}\Vert] ~|~ \mathscr{F}_{k-1}] \\
    & \quad\quad \quad \quad\quad \quad - \Vert\mb{C}(\mb{A}\mb{x}_{k-1} + \mb{B}\mb{u}_{k-1} + \mb{D}\mb{d}_{k-1} - \bs{\rho}\Vert \nonumber \\
    & \leq \Vert \mb{C}(\mb{A}\mb{x}_{k-1} + \mb{B}\mb{u}_{k-1}) - \bs{\rho} \Vert +  \mathbb{E}[ \Vert \mb{C}\mb{D}\mb{d}_{k-1} \Vert]] \\
    & \quad\quad \quad \quad - \Vert\mb{C}(\mb{A}\mb{x}_{k-1} + \mb{B}\mb{u}_{k-1}  - \bs{\rho})\Vert + \Vert  \mb{C}\mb{D}\mb{d}_{k-1}\Vert \nonumber \\
    & = \mb{E}[\Vert \mb{C} \mb{D} \mb{d}_{k-1} \Vert] + \Vert \mb{C} \mb{D} \mb{d}_{k-1}\Vert \\
    & \leq \mb{E}[\Vert \mb{C} \mb{D} \mb{d}_{k-1} \Vert] + d_\textup{max}\label{eq:E_on_ball}\\
    & = \int_{\mathcal{B}_{d_\textup{max}}(0)} \frac{1}{\pi d^2_\textup{max}} \Vert \mb{C} \mb{D} \mb{d}_{k-1} \Vert d\mb{d}_{k-1} + d_\textup{max}\\
    & = \int_{0}^{2\pi}\int_0^{d_\textup{max}}\frac{r^2}{\pi d^2_{\textup{max}}} dt d\theta + d_\textup{max}  = \frac{2}{3}d_\textup{max} + d_\textup{max} \\
    & = \frac{5}{3}d_\textup{max} \triangleq \delta. 
\end{align}
where we use the triangle inequality and then calculate the expectation in \eqref{eq:E_on_ball} exactly on the ball $\mathcal{B}_{d_\textup{max}}(0)\subset \R^2 $ using a coordinate transform to produce a less conservative bound. 
}

\revised{
Next, to bound $\sigma^2$ consider some constant vector $\mb{a} \in \R^2$ and random vector $\mb{b} \in \R^2$ that is uniformly distributed on the ball of radius $\mb{d}_\textup{max}$, i.e. $\mathcal{B}_{d_{\textup{max}}}(0)$. 
}

\revised{
First, we will lower-bound $\E[\Vert \mb{a} - \mb{b}\Vert]^2$ using Jensen's inequality given the convexity of the 2-norm: 
\begin{align}
    &\E [\Vert \mb{a} - \mb{b} \Vert ]
     \geq \Vert \mb{a} - \E[\mb{b}] \Vert = \Vert \mb{a} \Vert  \geq 0   \\
    &\implies \E[\Vert \mb{a} - \mb{b}\Vert]^2  \geq \Vert \mb{a} \Vert^2
\end{align}
}
\revised{
Next, we will bound $\E[\Vert \mb{a} - \mb{b}\Vert^2] $ by using the definition of the 2-norm squared and the linearity of the expectation, 
\begin{align}
    &\E[\Vert \mb{a} - \mb{b}\Vert^2 ] = \E[(\mb{a} - \mb{b})^\top (\mb{a} - \mb{b}) ] \\
    & = \E[ \mb{a}^\top \mb{a} - 2 \mb{a}^\top \mb{b} + \mb{b}^\top \mb{b}]\\
     & = \Vert \mb{a} \Vert^2 - 2 \mb{a}^\top \E[\mb{b}] + \E [ \Vert \mb{b} \Vert^2 ] = \Vert \mb{a} \Vert^2 + \E[\Vert \mb{b} \Vert^2]
\end{align}
We then use these two bounds, along with a coordinate transform, to calculate the variance. 
\begin{align}
    \textup{Var}&(\Vert \mb{a} - \mb{b} \Vert)  = \E[\Vert \mb{a} - \mb{b} \Vert^2] - \E[\Vert \mb{a} - \mb{b} \Vert]^2\\
    & \leq \E[\Vert \mb{a} - \mb{b}\Vert^2 ] - \Vert \mb{a} \Vert^2 =  \Vert \mb{a}\Vert^2  + \E[\Vert \mb{b}\Vert^2]  - \Vert \mb{a} \Vert^2 \nonumber \\
    & = \int_0^{2\pi}\int_{0}^{d_\textup{max}} \frac{r^3}{\pi d_\textup{max}^2} drd\theta = \frac{2\pi d_\textup{max}^4}{4 \pi d_\textup{max}^2} = \frac{1}{2}d_{\textup{max}}^2 \triangleq \sigma^2\nonumber 
\end{align}
}

\revised{
To find the value for $\sigma^2$ define $\mb{a} \triangleq \mb{C}(\mb{A}\mb{x}_{k-1} + \mb{B}(\mb{u}_{k-1}) - \bs{\rho}) $ and $\mb{b}  \triangleq \mb{C}\mb{D}\mb{d}_{k-1}$ where $\mb{d}_{k-1} \sim \mathcal{D}$ and note that variance is translationally invariant allowing us to reintroduce $r$ and set $\sigma^2 = \frac{1}{2}d_\textup{max}^2$. 
}

\revised{
Thus, given the structure of the example problem in Section \ref{example:hlip} we have found values $\delta = \frac{5}{3}d_\textup{max}$ and $\sigma^2 = \frac{d_\textup{max}^2}{2}$ which satisfy the conditions of Thm. \ref{thm:main}.
}

\ifthenelse{\boolean{long}}{
    }{
        \vspace{-0.5em}
    }
    
}{

    Proof of Proposition \ref{thm:comparison}:
\ifthenelse{\boolean{long}}{
    }{
        \vspace{-0.5em}
    }

}





\ifthenelse{\boolean{long}}{
\selectfont
}{
\selectfont
\vspace{-1.5em}
}
 
\bibliography{cosner}
\bibliographystyle{ieeetr}

\end{document}